\documentclass[a4paper,11pt]{article}
\usepackage[english]{babel}
\usepackage[utf8]{inputenc}
\usepackage[round]{natbib}
\usepackage[draft]{hyperref}

\usepackage[left=3cm, right=3cm, top=3cm, bottom=3cm]{geometry}
\usepackage{graphicx}
\usepackage{cancel}
\usepackage{amsmath,amsthm,amsfonts,amssymb}
\usepackage{mathrsfs} 
\usepackage{bm}
\usepackage{authblk}

\newtheorem{proposition}{\bf Proposition}
\newtheorem{theorem}{\bf Theorem}
\newtheorem{lemma}{\bf Lemma}

\newtheorem{corollary}{\bf Corollary}
\def\R{\mathbb{R}}

\newcommand{\angmom}{{\bm c}}
\newcommand{\energy}{{\cal E}}
\newcommand{\lenz}{{\bm L}}

\newcommand{\erre}{{\bm r}}
\newcommand{\erredot}{\dot{\bm r}}

\newcommand{\erho}{\hat{\bm e}^\rho}
\newcommand{\ealpha}{\hat{\bm e}^\alpha}
\newcommand{\edelta}{\hat{\bm e}^\delta}
\def\eperp{\hat{{\bm e}}^\perp}

\def\alphadot{\dot{\alpha}}
\def\deltadot{\dot{\delta}}
\def\rhodot{\dot{\rho}}
\def\bqdot{\dot{\bm{q}}}
\def\bq{\bm{q}}
\def\DD{{\bm D}}
\def\EE{{\bm E}}
\def\FF{{\bm F}}
\def\GG{{\bm G}}
\def\JJ{{\bm J}}
\def\NN{{\bm N}}
\def\OO{{\bm O}}
\def\PP{{\bm P}}

\def\Att{{\cal A}}
\def\Avec{{\cal D}}
\def\Svec{{\cal S}}
\def\Svectil{\widetilde{\cal S}}
\def\Vvectil{\widetilde{\cal V}}
\def\bPhi{\bm{\Phi}}
\def\bPsi{\bm{\Psi}}
\def\bzero{\bm{0}}

\def\D1xD2{{\bm W}_{12}}
\def\qtilde{\tilde{\mathfrak{q}}}

\title{Orbit determination from one position vector and a very short arc of optical observations}
\author[1]{Erica Scantamburlo\thanks{erica.scantamburlo@polito.it}}
\author[2]{Giovanni F. Gronchi\thanks{giovanni.federico.gronchi@unipi.it}}
\author[2]{Giulio Ba\`u\thanks{giulio.bau@unipi.it}}
\affil[1]{Department of Mechanical and Aerospace Engineering, Politecnico di Torino, Corso Duca degli Abruzzi 24, 10129, Italy}
\affil[2]{Dipartimento di Matematica, University of Pisa, Largo B. Pontecorvo 5, 56127, Italy}

\begin{document}

\maketitle

\begin{abstract}
  In this paper we address the problem of computing a preliminary
  orbit of a celestial body from one topocentric position vector and a
  very short arc (VSA) of optical observations. Using the conservation
  laws of the two-body dynamics, we write the problem as a system of 8
  polynomial equations in 6 unknowns. We prove that this system is
  generically consistent, namely it admits solutions at least in the
  complex field. From this system we derive a univariate polynomial
  $\mathfrak{v}$ of degree 8 in the unknown topocentric distance at
  the mean epoch of the VSA. Through Gr\"obner bases theory, we show
  that the degree of $\mathfrak{v}$ is minimum among the degrees of
  all the univariate polynomials solving this problem. The proposed
  method is relevant for different purposes, e.g. the computation of a
  preliminary orbit of an Earth satellite with radar and optical
  observations, the detection of maneuvres of an Earth satellite, and
  the recovery of asteroids which are lost due to a planetary close
  encounter. We also show some numerical tests in the case of
  asteroids undergoing a close encounter with the Earth.
\end{abstract}
\textbf{Keywords:} Orbit determination -- Keplerian integrals -- Algebraic methods

%---------------------
%---------------------
\section{Introduction}

The problem of computing the orbit of celestial bodies has attracted
the interest of scientists since a long time, see \citet{laplace1780},
\citet{lagrange}, \citet{gauss1809}. The recent improvements in the
observation technology have posed new interesting mathematical
problems in this field.  The number of asteroids observations
performed by modern telescopes is very large. Usually they can be
grouped into very short arcs (VSAs) of optical observations, however
it is not easy to determine whether VSAs collected in different nights
belong to the same observed objects.  In general, from a VSA we can
not compute a reliable preliminary orbit, but we can try to put
together different VSAs to perform this task, see
e.g. \citet{milanietal2005}. Assuming that two VSAs belong to the same
asteroid, we can write polynomial equations to compute a preliminary
orbit using the conservation laws of the two-body dynamics.
%This way of solving the orbit
%determination problem is called the {\em Keplerian integrals} method:
Recently, different ways to combine these integrals of motion have been
developed and tested, see \citet{th77,gdm10,gfd11,gbm15,gbm17}.

In this paper we address the problem of computing the orbit of a
celestial body (hereafter, the OD problem) from one topocentric
position $\mathcal{P}_1 = ( \alpha_1 ,\delta _1, \rho_1)$ at epoch
$t_1$ and a VSA of optical observations, from which we can derive an
attributable $\mathcal{A}_2 =
(\alpha_2,\delta_2,\dot{\alpha}_2,\dot{\delta}_2)$ at the mean epoch
$\overline{t}_2$ of the arc.
Here, $\alpha$, $\delta$, $\rho$ denote respectively right ascension,
declination and topocentric distance of the body, while $\alphadot,
\deltadot$ stand for the angular rates.

Using the conservation of angular momentum, energy and Laplace-Lenz
vector we write the OD problem as a system of polynomial equations in
the unknowns $\dot{\rho}_1$, $\dot{\alpha}_1$, $\dot{\delta}_1$,
$\rho_2$, $\dot{\rho}_2$, $z_2$, where $z_2$ is an auxiliary variable,
and the other variables allow us to obtain an orbit in spherical
coordinates at the two epochs.

We prove that this polynomial system is consistent, namely it
generically admits solutions, at least in the complex field, and we
obtain a univariate polynomial $\mathfrak{v}$ of degree 8 in the
unknown range $\rho_2$ to solve the OD problem. Through a computer
algebra software we are also able to show that the degree of
$\mathfrak{v}$ is minimum among the degrees of all the univariate
polynomials in $\rho_2$ solving this problem.

The proposed method is relevant for different purposes, such as the
computation of a preliminary orbit of an Earth satellite with radar
and optical observations, the detection of maneuvres of an Earth
satellite,
%\footnote{\textcolor{red}{Alcuni esempi di lavoro possono essere \cite{goff2015,escribano2022}}},
the recovery of asteroids which are lost due to a planetary encounter,
if the latter can be modeled as an instantaneous change of direction
of the velocity vector like in \"Opik theory \citep{opik1976}. Here we
show the results of the application of our algorithm to the orbits of
some near-Earth asteroids, whose positions have been changed to
increase the effect of the close encounter with the Earth.

The paper is organized as follows: after recalling the expressions of
the Keplerian integrals in Section~\ref{s:kepint}, we introduce the OD
problem in Section~\ref{s:problem} as an overdetermined polynomial
system, whose consistency is shown in Section~\ref{s:consistent},
where the univariate polynomial $\mathfrak{v}$ is derived. The
minimality of the degree of $\mathfrak{v}$ is proved in
Section~\ref{s:groebner}. In Section~\ref{s:orbselect} we discuss the
selection of the solutions. Finally, in Section~\ref{s:numerics} we
present the numerical tests.

%----------------------------
%----------------------------
\section{Keplerian integrals}
\label{s:kepint}

Let us consider a celestial body whose dynamics can be modeled by
Kepler's problem
\begin{equation}
  \ddot{\erre} = -\mu\displaystyle\frac{\erre}{|\erre|^3},
  \label{kepeq}
\end{equation}
where $\erre$ and $\mu$ denote respectively its position and the
gravitational parameter.
%In this context are included e.g. the cases of a satellite orbiting
%the Earth, or an asteroid orbiting the Sun. Hereafter, we shall speak
%of the case of an Earth satellite.
Equation~\eqref{kepeq} admits the first integrals
\begin{equation}
  \angmom = \erre \times \dot{\erre},\hspace{0.6cm}
  \energy = \frac{1}{2}|\erredot|^2 - \frac{\mu}{|\erre|},\hspace{0.6cm}
  \lenz = \frac{1}{\mu}\erredot\times\angmom -\frac{\erre}{|\erre|},
  \label{integrals}
\end{equation}
corresponding to the angular momentum, the energy, and the
Laplace-Lenz vector of the body, respectively. We call {\em Keplerian
  integrals} these constants of motion.

\noindent Note that we can write
\[
\mu\lenz = \Bigl(|\erredot|^2 - \frac{\mu}{|\erre|}\Bigr)\erre -
(\erredot\cdot\erre)\erredot.
\]

%-----------------------
%-----------------------
\section{The OD problem}
\label{s:problem}

Let us introduce a reference frame whose origin is the center of force
of the Keplerian dynamics. We consider the Keplerian motion of a
celestial body around the origin observed from a moving point of view.
Assume the position $\bq$ and the velocity $\bqdot$ of the observer
are known functions of time. The position of the observed body is
given by
\begin{equation}
  \erre = \bq+\rho\erho,
  \label{position}
\end{equation}
where $\rho$ represents the topocentric distance, and
\[
\erho = (\cos\delta\cos\alpha, \cos\delta\sin\alpha, \sin\delta)
\]
is the {\em line of sight}.
%with $\alpha$ and $\delta$ the right
%ascension and declination of the body.

Let us assume that we know the position vector
\[
{\cal P}_1 = (\alpha_1,\delta_1,\rho_1) \in [-\pi,\pi) \times (-\pi/2,\pi/2) \times \mathbb{R}^{+} 
\]
at epoch $t_1$, and a very short arc of optical observations of the
same object, from which we compute the attributable
\[
{\cal A}_2 = (\alpha_2,\delta_2,\alphadot_2,\deltadot_2) \in [-\pi, \pi) \times (-\pi/2,\pi/2)
\times \mathbb{R}^{2}
\]
at the mean epoch $\bar{t}_2$. We can combine these data to compute
the quantities $\rhodot_1$, $\xi_1$, $\zeta_1$, $\rho_2$, $\rhodot_2$,
%\in \mathbb{R}^{3} \times \mathbb{R}^{+} \times \mathbb{R} $,
where
\[
\xi_1 = \rho_1\alphadot_1\cos{\delta_1},\qquad
\zeta_1 = \rho_1\deltadot_1,
\]
which are missing to have a complete set of orbital elements at both
epochs. For this purpose, we use the conservation of the integrals
listed in~\eqref{integrals}.

\noindent Hereafter, we shall use subscripts $1, 2$ for all the
quantities relative to epochs $t_1$, $\bar{t}_2$.

\noindent The dependence of position and velocity on the unknowns is
given by
%the following relations:
\begin{equation}
  \begin{split}
    \erredot_1 &= \bqdot_1 + \rhodot_1\erho_1 + \xi_1\ealpha_1
    +\zeta_1 \edelta_1,\cr
    \erre_2 &= \bq_2 + \rho_2\erho_2,\cr
    \erredot_2 &= \bqdot_2 + \rhodot_2\erho_2 + \rho_2\eperp_2,
  \end{split}
  \label{velocity}
\end{equation}
where
\[
\eperp_2 = \alphadot_2\cos\delta _2\ealpha_2 + \deltadot_2\edelta _2
\]
is a known vector. Note that $\erre_1 = \bq_1+\rho_1\erho_1$ is
known.

To write a polynomial system, we replace the term
$\frac{\mu}{|\erre_2|}$ with an auxiliary variable $z_2$, independent
from $\rho_2$, and set
\[
\begin{aligned}
  \mu\widetilde{\lenz}_2 &= \Bigl(|\erredot_2|^2 - z_2\Bigr)\erre_2 -
  (\erredot_2\cdot\erre_2)\erredot_2,\\[0.5ex]
  \widetilde{\energy}_2 &= \frac{1}{2}|\erredot_2|^2 - z_2.
\end{aligned}
\]
\begin{lemma}
  The following relations hold:
  \begin{equation}
    \mu^2|\lenz_1|^2 = 2\energy_1|\angmom_1|^2 + \mu^2,\qquad
    \mu^2|\widetilde{\lenz}_2|^2 = 2\widetilde{\energy}_2|\angmom_2|^2 + z_2^2|\erre_2|^2.
    \label{cEL_rels}
  \end{equation}
  \label{l:cEL}
\end{lemma}
\begin{proof}
  From the definitions of $\lenz_1$ and $\widetilde{\lenz}_2$
  %(see equation \eqref{integrals}$_3$)
  we have
  \begin{align*}
    \mu^2|\lenz_1|^2 &= |\erredot_1\times\angmom_1|^2 - 2\frac{\mu}{|\erre_1|}(\erredot_1\times\angmom_1\cdot\erre_1)
    + \mu^2,\\[0.5ex]
    \mu^2|\widetilde{\lenz}_2|^2 &= |\erredot_2\times\angmom_2|^2 - 2z_2(\erredot_2\times\angmom_2\cdot\erre_2)
    + z_2^2|\erre_2|^2.
  \end{align*}
  We note that, at both epochs,
  \[
  |\erredot\times\angmom|^2 = |\erredot|^2|\angmom|^2,\qquad
  \erredot\times\angmom\cdot\erre = |\angmom|^2,
  \]
  so that
  \begin{align*}
    \mu^2|\lenz_1|^2 &= |\angmom_1|^2\Bigl(|\erredot_1|^2 - 2\frac{\mu}{|\erre_1|}\Bigr) + \mu^2,\\[0.5ex]
    \mu^2|\widetilde{\lenz}_2|^2 &= |\angmom_2|^2\bigl(|\erredot_2|^2 - 2z_2\bigr) + z_2^2|\erre_2|^2.
  \end{align*}
  Relations~\eqref{cEL_rels} immediately follow from the definitions
  of $\energy_1$ and $\widetilde{\energy}_2$.
\end{proof}

We consider the overdetermined polynomial system
\begin{equation}
  \angmom_1 = \angmom_2,\quad
  \lenz_1 = \widetilde{\lenz}_2,\quad
  \energy_1 = \widetilde{\energy}_2,\quad
  z_2^2|\erre_2|^2 = \mu^2,
  \label{fullsys}
\end{equation}
consisting of 8 equations in the 6 unknowns $\rhodot_1$, $\xi_1$,
$\zeta_1$, $\rho_2$, $\rhodot_2$, $z_2$. System~\eqref{fullsys}
corresponds to the equations of our OD problem.

%-------------------------------------
%-------------------------------------
\section{Consistency of the equations}
\label{s:consistent}

We will show the following property:

\begin{theorem}
  The overdetermined polynomial system~\eqref{fullsys} is generically
  consistent, i.e. it always has solutions in the complex field for a
  generic choice of the data ${\cal P}_1$, ${\cal A}_2$, $\bq_1$,
  $\bqdot_1$, $\bq_2$, $\bqdot_2$.
  \label{teo:consist}
\end{theorem}

\noindent The proof makes use of the results presented in the
following subsections.

We start by noting that, because of relations~\eqref{cEL_rels},
equation
\[
z_2^2|\erre_2|^2 = \mu^2
\]
is a consequence of the reduced system
\begin{equation}
  \angmom_1 = \angmom_2,\quad
  \lenz_1 = \widetilde{\lenz}_2,\quad
  \energy_1 = \widetilde{\energy}_2,
  \label{redsys}
\end{equation}
of 7 equations in 6 unknowns. Therefore, to prove
Theorem~\ref{teo:consist} we show that system~\eqref{redsys} is
generically consistent.

Moreover, for a generical choice of the data ${\cal P}_1$, ${\cal
  A}_2$, $\bq_1$, $\bqdot_1$, $\bq_2$, $\bqdot_2$,
system~\eqref{redsys} is equivalent to
\[
\mathfrak{q}_1 = \mathfrak{q}_2 = \mathfrak{q}_3 = \mathfrak{q}_4 =
\mathfrak{q}_5 = \mathfrak{q}_6 = \mathfrak{q}_7= 0,
\]
where
\begin{equation}
  \begin{split}
    \mathfrak{q}_1 &= (\bm{c}_1 - \bm{c}_2)\cdot\D1xD2,\cr
    \mathfrak{q}_2 &= (\bm{c}_1 - \bm{c}_2)\cdot\bm{D}_1\times\D1xD2,\cr
    \mathfrak{q}_3 &= (\bm{c}_1 - \bm{c}_2)\cdot\bm{D}_2\times\D1xD2,\cr
    \mathfrak{q}_4 &= \mu(\bm{L}_1 - \tilde{\bm{L}}_2)\cdot\bm{D}_1,\cr
    \mathfrak{q}_5 &= \mu(\bm{L}_1 - \tilde{\bm{L}}_2)\cdot\bm{D}_2,\cr
    \mathfrak{q}_6 &= \mu(\bm{L}_1 - \tilde{\bm{L}}_2)\cdot(\bm{r}_1\times\hat{\bm{e}}^{\rho}_2),\cr
    \mathfrak{q}_7 &= \energy_1 - \tilde{\energy}_2,
  \end{split}
  \label{qgenerators}
\end{equation}
with
\begin{equation}
  \DD_j = \bq_j\times\erho_j,\quad j=1,2,
  \label{DDj}
\end{equation}
and
\begin{equation}
  \D1xD2 =  \DD_1\times\DD_2.
  \label{D1xD2}
\end{equation}

%--------------------------------
\subsection{The angular momentum}

The angular momenta $\angmom_1$ and $\angmom_2$, written in terms of
the unknowns,
%$\rhodot_1,\xi_1,\zeta_1$ and $\rhodot_2,\rho_2$
become
\begin{align*}
  \angmom_1 &= \DD_1\rhodot_1 + \NN_1\xi_1 + \OO_1\zeta_1 + \PP_1,\\[0.5ex]
  \angmom_2 &= \DD_2\rhodot_2 +\EE_2\rho_2 ^2 +\FF_2\rho_2 + \GG_2,
\end{align*}
where $\DD_1$, $\DD_2$ are defined as in~\eqref{DDj}, and
\begin{align*}
  %\DD_1 &= \bq_1 \times \erho_1,   &\DD_2 &= \bq_2 \times \erho_2,\\
  \NN_1 &= \erre_1\times\ealpha_1, & \EE_2 &= \alphadot_2\cos\delta_2\edelta_2 - \deltadot_2\ealpha_2,\\
  \OO_1 &= \erre_1\times\edelta_1, & \FF_2 &= \alphadot_2\cos\delta_2(\bq_2\times\ealpha_2) +
  \deltadot_2(\bq_2\times\edelta_2) + (\erho_2\times\bqdot_2),\\
  \PP_1 &= \erre_1\times\bqdot_1, & \GG_2 &= \bq_2\times\bqdot_2.
\end{align*}
Therefore, we can write 
\begin{equation}
  \angmom_1 - \angmom_2 = \DD_1\rhodot_1 - \DD_2\rhodot_2 + \JJ(\xi_1,\zeta_1,\rho_2),
  \label{angmomdiff}
\end{equation}
with
\begin{equation}
  \JJ(\xi_1,\zeta_1,\rho_2) = \NN_1\xi_1 + \OO_1\zeta_1 - \EE_2\rho_2^2 - \FF_2\rho_2 + \PP_1 - \GG_2.
  \label{JJ}
\end{equation}

%------------------------------------
\subsection{Elimination of variables}

Using relations~\eqref{angmomdiff},~\eqref{JJ}, the polynomials
$\mathfrak{q}_1$, $\mathfrak{q}_2$, $\mathfrak{q}_3$ defined in
(\ref{qgenerators}) can be written as
\begin{equation*}
  \begin{aligned}
    \mathfrak{q}_1 &= Q^{(1)}_{100}\xi_1 + Q^{(1)}_{010}\zeta_1 + Q^{(1)}_{002}\rho_2^2 + Q^{(1)}_{001}\rho_2 + Q^{(1)}_{000},\\[0.5ex]
    \mathfrak{q}_2 &= -|\D1xD2|^2\rhodot_2 + Q^{(2)}_{100}\xi_1 + Q^{(2)}_{010}\zeta_1 + Q^{(2)}_{002}\rho_2^2 + Q^{(2)}_{001}\rho_2
    + Q^{(2)}_{000},\\[0.5ex]
    \mathfrak{q}_3 &= |\D1xD2|^2\rhodot_1 + Q^{(3)}_{100}\xi_1 + Q^{(3)}_{010}\zeta_1 + Q^{(3)}_{002}\rho_2^2 + Q^{(3)}_{001}\rho_2
    + Q^{(3)}_{000},
    % \mathfrak{q}_5 &= f_5(\rhodot_1,\rhodot_2,\xi_1,\zeta_1,\rho_2),\cr
    % \mathfrak{q}_6 &= f_6(\rhodot_1,\rhodot_2,z_2,\xi_1,\zeta_1,\rho_2),\cr
    % \mathfrak{q}_7 &= z_2 + f_7(\rhodot_1,\rhodot_2,\xi_1,\zeta_1,\rho_2)
  \end{aligned}
  %\label{qgen2}
\end{equation*}
where
\begin{equation*}
  \begin{aligned}
    Q^{(1)}_{100} &= \NN_1\cdot\D1xD2, & Q^{(1)}_{010} &= \OO_1\cdot\D1xD2,\\[0.5ex]
    Q^{(1)}_{002} &= -\EE_2\cdot\D1xD2, & Q^{(1)}_{001} &= -\FF_2\cdot\D1xD2,\\[0.5ex]
    Q^{(1)}_{000} &= (\GG_2-\PP_1)\cdot\D1xD2, & &
  \end{aligned}
\end{equation*}
\begin{equation*}
  \begin{aligned}
    Q^{(2)}_{100} &= \NN_1\cdot\DD_1\times\D1xD2, & Q^{(2)}_{010} &= \OO_1\cdot\DD_1\times\D1xD2,\\[0.5ex]
    Q^{(2)}_{002} &= -\EE_2\cdot\DD_1\times\D1xD2, & Q^{(2)}_{001} &= -\FF_2\cdot\DD_1\times\D1xD2,\\[0.5ex]
    Q^{(2)}_{000} &= (\PP_1-\GG_2)\cdot\DD_1\times\D1xD2, & &
  \end{aligned}
\end{equation*}
\begin{equation*}
  \begin{aligned}
    Q^{(3)}_{100} &= \NN_1\cdot\DD_2\times\D1xD2, & Q^{(3)}_{010} &= \OO_1\cdot\DD_2\times\D1xD2,\\[0.5ex]
    Q^{(3)}_{002} &= -\EE_2\cdot\DD_2\times\D1xD2, & Q^{(3)} _{001} &= -\FF_2\cdot\DD_2\times\D1xD2,\\[0.5ex]
    Q^{(3)}_{000} &= (\PP_1-\GG_2)\cdot\DD_2\times\D1xD2. & &  
  \end{aligned}
\end{equation*}
%and $f_5,f_6,f_7$ are polynomials.
In particular, the occurence of the variables $\rhodot_1$, $\xi_1$,
$\zeta_1$, $\rhodot_2$ in $\mathfrak{q}_1$,
$\mathfrak{q}_2$,$\mathfrak{q}_3$ is at most linear. From the
definitions above, we obtain
\begin{align*}
  Q^{(1)}_{100} &= (\NN_1\times\DD_1)\cdot\DD_2 = \left[(\erre_1\times\ealpha_1)\times(\erre_1\times\erho_1)\right]\cdot\DD_2\\
  &= \left[(\erre_1\times\ealpha_1)\cdot\erho_1\right]\erre_1\cdot\DD_2 = -(\erre_1\cdot\edelta_1)(\erre_1\cdot\DD_2),\\[0.5ex]
  Q^{(1)}_{010} &= (\OO_1\times\DD_1)\cdot\DD_2 = \left[(\erre_1\times\edelta_1)\times(\erre_1\times\erho_1)\right]\cdot\DD_2\\
  &= \left[(\erre_1\times\edelta_1)\cdot\erho_1\right]\erre_1\cdot\DD_2 = (\erre_1\cdot\ealpha_1)(\erre_1\cdot\DD_2).
\end{align*}
Therefore, the coefficients $Q_{100}^{(1)}$ and $Q_{010}^{(1)}$ are
both vanishing iff
\[
\erre_1\cdot\DD_2 = 0
\qquad\mbox{or}\qquad
\DD_1 = \bzero.
\]
In fact, the second condition in the alternative above is equivalent
to
\[
\erre_1\cdot\ealpha_1 = \erre_1\cdot\edelta_1 = 0.
\]
Thus, for a generic choice of the data, the angular momentum equations
allow us to eliminate $\rhodot_1$, $\rhodot_2$ and one variable
between $\xi_1$ and $\zeta_1$.
%by elementary linear algebra methods.
Assuming $Q^{(1)}_{100}\neq 0$, we choose to eliminate $\xi_1$, which
can be written as
\begin{equation}
  \xi_1 = - \frac{Q^{(1)}_{010}\zeta_1 + \sum_{h=0}^2 Q^{(1)}_{00h}\rho_2^h}{Q^{(1)}_{100}}.
  \label{xi1_explicit}
\end{equation}
Substituting~\eqref{xi1_explicit} in equations
$\mathfrak{q}_3=\mathfrak{q}_2 = 0$ and assuming $\D1xD2\neq \bzero$
we find
\begin{align}
  \dot{\rho}_1 &= - \frac{1}{|\D1xD2|^2Q^{(1)}_{100}}\Bigl[\left(Q^{(3)}_{010}Q^{(1)}_{100} - Q^{(1)}_{010}Q^{(3)}_{100}\right)\zeta_1
    + \sum_{h=0}^2\left(Q^{(3)}_{00h}Q^{(1)}_{100} - Q^{(1)}_{00h}Q^{(3)}_{100}\right)\rho_2^h\Bigr]
  \label{rhodot1}\\[0.5ex]
  \dot{\rho}_2 &= \frac{1}{|\D1xD2|^2 Q^{(1)}_{100}}\Bigl[\left(Q^{(2)}_{010}Q^{(1)}_{100} - Q^{(1)}_{010}Q^{(2)}_{100}\right)\zeta_1
    + \sum_{h=0}^2\left(Q^{(2)}_{00h}Q^{(1)}_{100} - Q^{(1)}_{00h}Q^{(2)}_{100}\right)\rho_2^h\Bigr].
  \label{rhodot2}
\end{align}
Hence, by using relations~\eqref{xi1_explicit},~\eqref{rhodot1},
and~\eqref{rhodot2} we can eliminate the variables $\xi_1$,
$\rhodot_1$, $\rhodot_2$ in the generators $\mathfrak{q}_4$,
$\mathfrak{q}_5$, $\mathfrak{q}_6$, $\mathfrak{q}_7$. The resulting
polynomials, named $\qtilde_4$, $\qtilde_5$, $\qtilde_6$, $\qtilde_7$,
can be written as
\begin{align*}
  \qtilde_4 &= (\erre_2\cdot\DD_1)z_2 + h_4(\zeta_1,\rho_2),\\[0.5ex]
  \qtilde_5 &= \qtilde_5 (\zeta_1,\rho_2),\\[0.5ex]
  \qtilde_6 &= -(\erre_1\cdot\DD_2) z_2 + h_6(\zeta_1,\rho_2),\\[0.5ex]
  \qtilde_7 &= z_2 + h_7 (\zeta_1,\rho_2),
\end{align*}
for some bivariate polynomials $h_4$, $h_6$, $h_7$.
%
%Finally we have
%\[
%\mathfrak{q}_7 = z_2 + \frac{1}{2}(\Vert\erredot_1\Vert^2 - \Vert\erredot_2\Vert^2)
%- \frac{\mu}{\Vert\erre_1\Vert}.
%\]
% 
%We observe that the generator $\mathfrak{q}_1$ is linear in $\xi_1$
%and $\zeta_1$, while it is quadratic in $\rho_2$; hence, If $Q^{(1)}
%_{100} \neq 0$ (or $Q^{(1)} _{010} \neq 0$)\footnote{$Q^{(1)} _{100} =
%  0$ when $\bq_1 \cdot \edelta_1 = 0$ or $\erre_1 \cdot \DD_2 =0$,
%  while $Q^{(1)} _{010} =0$ when $\bq_1 \cdot \ealpha_1 =0$ or
%  $\erre_1 \cdot \DD_2 = 0$.} we can use relation $\mathfrak{q}_1 = 0$
%to eliminate $\xi_1$ (or $\zeta_1$) from the other equations in
%\eqref{polysys}.
%
%Noting that
%\[
%\angmom_2\cdot\erho_2 = (\FF_2\cdot\erho_2)\rho_2 + \GG_2\cdot\erho_2,
%\]
%and using relations \eqref{xi1_explicit}, \eqref{rhodot1},
%\eqref{rhodot2} we can write $\tilde{\mathfrak{q}}_4$,
%$\tilde{\mathfrak{q}}_5$, $\tilde{\mathfrak{q}}_6$, $\mathfrak{q}_7$,
%as follows:
%
%
\noindent For later use we observe that
\[
\qtilde_5 =  P^{(5)}_{12}\zeta_1\rho_2^2 + P^{(5)}_{11}\zeta_1\rho_2 + P^{(5)}_{10}\zeta_1 + P^{(5)}_{04}\rho_2^4 +
P^{(5)}_{03}\rho_2^3 + P^{(5)}_{02}\rho_2^2 + P^{(5)}_{01}\rho_2 + P^{(5)}_{00},
\]
for some coefficients $P^{(5)}_{ij}$. Moreover, using equation
$\qtilde_7 = 0$, we can eliminate $z_2$ from $\qtilde_4$ and
$\qtilde_6$. In particular, with this elimination $\qtilde_6$ becomes
\[
\mathfrak{p}_6 = P^{(6)}_{20}\zeta_1^2 + P^{(6)}_{12}\zeta_1\rho_2^2 + P^{(6)}_{11}\zeta_1\rho_2 + P^{(6)}_{10}\zeta_1
+ P^{(6)}_{04}\rho_2^4 + P^{(6)}_{03}\rho_2^3 + P^{(6)}_{02}\rho_2^2 + P^{(6)}_{01}\rho_2 + P^{(6)}_{00},
\]
for some coefficients $P^{(6)}_{ij}$.

%--------------------------------
\subsection{A redundant relation}

Here, we prove that equation $\mathfrak{q}_4=0$ can be dropped from
system~\eqref{redsys} without losing any solution. In particular, we
show:
\begin{proposition}
  If $\erre_1\cdot\DD_2\neq 0$, the polynomial $\mathfrak{q}_4$ is
  generated by
  \[
  \mathfrak{q}_1,\,\mathfrak{q}_2,\,\mathfrak{q}_3,\,\mathfrak{q}_6.
  \]
  \label{prop:redundant}
\end{proposition}

\begin{proof}
  Let us consider the polynomials $\tilde{\mathfrak{q}}_4$,
  $\tilde{\mathfrak{q}}_6$ obtained from $\mathfrak{q}_4$,
  $\mathfrak{q}_6$ by eliminating $\rhodot_1$, $\rhodot_2$, $\xi_1$
  through relations $\mathfrak{q}_1 = \mathfrak{q}_2 = \mathfrak{q}_3
  = 0$. We note that
  \[
  \tilde{\mathfrak{q}}_4 = \sum_{j=1}^3 \mathfrak{a}_j\mathfrak{q}_j + \mathfrak{q}_4,\qquad
  \tilde{\mathfrak{q}}_6 = \sum_{j=1}^3 \mathfrak{b}_j\mathfrak{q}_j + \mathfrak{q}_6,
  \]
  for some polynomials $\mathfrak{a}_j$, $\mathfrak{b}_j$. % $j=1,2,3$.
  We prove that $\tilde{\mathfrak{q}}_6$ divides $\tilde{\mathfrak{q}}_4$; in particular,
  \begin{equation}
    \tilde{\mathfrak{q}}_4 = -\frac{\erre_2\cdot\DD_1}{\erre_1\cdot\DD_2}\,\tilde{\mathfrak{q}}_6.
    \label{q4q6dep}
  \end{equation}
  For this purpose, we first note that system $\mathfrak{q}_1 =
  \mathfrak{q}_2 = \mathfrak{q}_3 = 0$ is equivalent to $\angmom_1 =
  \angmom_2$ and, inserting the last relation in the expressions of
  $\mathfrak{q}_4$, $\mathfrak{q}_6$, after eliminating $\rhodot_1$,
  $\rhodot_2$, $\xi_1$, we obtain
  \begin{align*}
    \tilde{\mathfrak{q}}_4 &= (\angmom_1\cdot\erho_1)\left[\erre_1\cdot(\erredot_1-\erredot_2)\right]
    + z_2(\erre_2\cdot\DD_1),\\[0.5ex]
    %\tilde{\mathfrak{q}}_5 &= (\angmom_2\cdot\erho_2)\left[\erre_2\cdot(\erredot_1-\erredot_2)\right]
    %- (\erre_1\cdot\DD_2)\frac{\mu}{|\erre_1|}\\[0.5ex]
    \tilde{\mathfrak{q}}_6 &= (\angmom_2\cdot\erho_2)\left[\erre_1\cdot(\erredot_1-\erredot_2)\right]
    - z_2(\erre_1\cdot\DD_2),
  \end{align*}
  where $\angmom_1$, $\angmom_2$, $\erre_2$, $\erredot_1$,
  $\erredot_2$ are meant as functions of $\zeta_1$, $\rho_2$ only.
  More details about these computations can be found in
  Appendix~\ref{app:q5q6}.

  We have
  \[
  \erre_2\cdot\DD_1 = -(\erre_1\times\erre_2)\cdot\erho_1,\qquad
  \erre_1\cdot\DD_2 = (\erre_1\times\erre_2)\cdot\erho_2.
  \]
  Moreover, $\angmom_1 = \angmom_2$ implies that $\erre_1 \times
  \erre_2$ is parallel to $\angmom_1$ and $\angmom_2$. Hence
  \[
  (\angmom_2\cdot\erho_2)\left[(\erre_1\times\erre_2)\cdot\erho_1\right] =
  (\angmom_1\cdot\erho_1)\left[(\erre_1\times\erre_2)\cdot\erho_2\right],
  \]
  that is
  \[
  -(\angmom_2\cdot\erho_2)(\erre_2\cdot\DD_1) = (\angmom_1\cdot\erho_1)(\erre_1\cdot\DD_2).
  \]
  The latter relation immediately yields~\eqref{q4q6dep}.

  Setting
  \[
  A = -\frac{\erre_2\cdot\DD_1}{\erre_1\cdot\DD_2},
  \]
  we can write
  \[
  \mathfrak{q}_4 = \tilde{\mathfrak{q}}_4 - \sum_{j=1}^3 \mathfrak{a}_j\mathfrak{q}_j =
  A\tilde{\mathfrak{q}}_6 - \sum_{j=1}^3\mathfrak{a}_j\mathfrak{q}_j =
  \sum_{j=1}^3(A\mathfrak{b}_j-\mathfrak{a}_j)\mathfrak{q}_j + A\mathfrak{q}_6,
  \]
  which concludes the proof.  
\end{proof}

%-------------------------------------
\subsection{The univariate polynomial}

We compute a univariate polynomial $\mathfrak{u}(\rho_2)$ of degree 8,
which is a consequence of the equations in~\eqref{redsys}. Let us
consider the polynomial system
\begin{equation}
  \mathfrak{q}_1 =
  \mathfrak{q}_2 =
  \mathfrak{q}_3 =
  \mathfrak{q}_5 =
  \mathfrak{q}_6 =
  \mathfrak{q}_7 = 0
  \label{polysys}
\end{equation}
including all the polynomials in~\eqref{qgenerators} except
$\mathfrak{q}_4$. To solve the OD problem, we solve the polynomial
system defined by (\ref{polysys}).

Now, we compute the resultant of $\tilde{\mathfrak{q}}_5$ and
$\mathfrak{p}_6$ with respect to $\zeta_1$, which is denoted by
$\mathfrak{v}(\rho_2)$. For this purpose, the terms in $\qtilde_5$ and
$\mathfrak{p}_6$ are grouped in the following way:
\begin{align*}
  \qtilde_5 &= a_1 (\rho_2)\zeta_1 + a_0(\rho_2),\\[0.5ex]
  \mathfrak{p}_6 &= P^{(6)}_{20}\zeta_1^2 + b_1(\rho_2)\zeta_1 + b_0(\rho_2),
\end{align*}
where
\begin{align*}
  a_1(\rho_2) &= P^{(5)}_{12}\rho_2^2 + P^{(5)}_{11}\rho_2 + P^{(5)}_{10},\\[0.5ex]
  a_0(\rho_2) &= P^{(5)}_{04}\rho_2^4 + P^{(5)}_{03}\rho_2^3 + P^{(5)}_{02}\rho_2^2 + P^{(5)}_{01}\rho_2 + P^{(5)}_{00},\\[0.5ex]
  b_1(\rho_2) &= P^{(6)}_{12}\rho_2^2 + P^{(6)}_{11}\rho_2 + P^{(6)}_{10},\\[0.5ex]
  b_0(\rho_2) &= P^{(6)}_{04}\rho_2^4 + P^{(6)}_{03}\rho_2^3 + P^{(6)}_{02}\rho_2^2 + P^{(6)}_{01}\rho_2 + P^{(6)}_{00}.
\end{align*}
Therefore, we have
\[
\mathfrak{v}(\rho_2) = a_1(\rho_2)a_0(\rho_2)b_1(\rho_2) - a_0^2(\rho_2)P^{(6)}_{20} - b_0(\rho_2)a_1^2(\rho_2),
\]
which has degree 8.

\medbreak We can conclude the proof of Theorem~\ref{teo:consist}. In
fact, $\mathfrak{v}(\rho_2)$ has (generically) 8 complex roots. For
each of these roots, equation $\tilde{\mathfrak{q}}_5 = 0$ gives a
unique value of $\zeta_1$. Then, equation $\tilde{\mathfrak{q}}_7 = 0$
and relations~\eqref{xi1_explicit},~\eqref{rhodot1},~\eqref{rhodot2},
respectively, allow us to compute unique values for the remaining
components $z_2$, $\xi_1$, $\rhodot_1$, $\rhodot_2$ of the solutions.

%----------------------------
\section{An optimal property}
% of $\mathfrak{v}$}
%{A Gr\"obner basis for the ideal $I$}
\label{s:groebner}

We show that the univariate polynomial $\mathfrak{v}(\rho_2)$ has the
minimal degree among all the univariate polynomials in the variable
$\rho_2$ that are algebraic consequences of
\[
\angmom_1 - \angmom_2,\
\mu(\lenz_1 - \widetilde{\lenz}_2),\
\energy_1 - \widetilde{\energy}_2.
\]
More precisely,
%The univariate polynomial $\mathfrak{v}(\rho_2)$ of degree 8 given by
%\eqref{eq:univarpolyrho2} belongs to
let us introduce the polynomial ideal
\[
I = \langle\angmom_1-\angmom_2,\
\mu(\lenz_1-\widetilde{\lenz}_2),\
\energy_1-\widetilde{\energy}_2\rangle
\subseteq\R[\rhodot_1,\xi_1,\zeta_1,\rho_2,\rhodot_2,z_2].
\]
We will show the following result:
\begin{theorem}
  For a generic choice of the data we can find a Gr\"obner basis
  $\mathcal{G}$ in $\R[\rhodot_1,\xi_1,\zeta_1,\rho_2,\rhodot_2,z_2]$
  of the ideal $I$ for the lexicographic order
  \[
  \rhodot_1\succ\rhodot_2\succ\xi_1\succ z_2\succ\zeta_1\succ\rho_2.
  \]
  such that $\mathfrak{v}\in\mathcal{G}$.
\end{theorem}

\begin{proof}
  As a consequence of Proposition~\ref{prop:redundant} we obtain
  \[
  I = \langle\mathfrak{q}_1,\mathfrak{q}_2,\mathfrak{q}_3,\mathfrak{q}_5,
  \mathfrak{q}_6,\mathfrak{q}_7\rangle,
  \]
  where we dropped $\mathfrak{q}_4$ from the set of generators. Let us
  consider the elimination ideal
  \[
  I' = I\cap\R[\zeta_1,\rho_2] =
  \langle\qtilde_5,\mathfrak{p}_6\rangle.
  \]
  We compute a Gr\"obner basis $\mathcal{G}'$ of the ideal $I'$ for the
  lexicographic order $ \zeta_1 \succ \rho_2$ using the software
  \emph{Mathematica}\footnote{Wolfram Research, Inc. Mathematica,
    Version 12.1, \tt{https://www.wolfram.com/mathematica}.}: in this
  basis we have a univariate polynomial
  \[
  \mathfrak{u}(\rho_2) = \sum_{i=0}^8 u_i\rho_2^i,
  \]
  for some constant coefficients $u_i$. The roots of $\mathfrak{u}$ are
  the only values of the $\rho_2$ components of the solutions of
  system~\eqref{polysys}.
  
  We note that, since the resultant $\mathfrak{v}$ is a univariate
  polynomial of degree 8, $\mathfrak{v}$ is necessarily proportional to
  $\mathfrak{u}$, i.e.
  \[
  \mathfrak{v} = \kappa\mathfrak{u},\qquad
  \kappa\in\mathbb{R}\backslash\{0\}.
  \]
  In place of the generators $\mathfrak{q}_1$, $\mathfrak{q}_2$,
  $\mathfrak{q}_3$, we can consider
  \begin{align*}
    \tilde{\mathfrak{q}}_1 &= \xi_1 + h_1(\zeta_1,\rho_2),\\[0.5ex]
    \tilde{\mathfrak{q}}_2 &= \rhodot_2 + h_2(\zeta_1,\rho_2),\\[0.5ex]
    \tilde{\mathfrak{q}}_3 &= \rhodot_1 + h_3(\zeta_1,\rho_2),
  \end{align*}
  where the polynomials $h_1$, $h_2$, $h_3$ are defined by
  relations~\eqref{xi1_explicit},~\eqref{rhodot1},~\eqref{rhodot2}.

  The set
  \[
  \mathcal{G} = \{\tilde{\mathfrak{q}}_1,\tilde{\mathfrak{q}}_2,\tilde{\mathfrak{q}}_3,\tilde{\mathfrak{q}}_7\}
  \cup\mathcal{G}'
  \]
  is a Gr\"obner basis of the ideal $I$. In fact, the leading term of
  every polynomial in $I$ is divisible by the leading term of one
  polynomial in $\mathcal{G}$.

\end{proof}

\noindent As a consequence of the previous theorem we obtain the
following property:
\begin{corollary}
  The polynomial $\mathfrak{v}$ has the minimum degree among the
  univariate polynomials in $\rho_2$ belonging to $I$.
\end{corollary}

%--------------------------------
%--------------------------------
\section{Selecting the solutions}
\label{s:orbselect}

From the solutions of~\eqref{redsys}, we first discard the ones with
non-real components and the ones with $\rho_2, z_2\leq 0$.
From the remaining solutions we can compute Keplerian orbital elements
at epochs
\[
\tilde{t}_1 = t_1-\frac{\rho_1}{c},\qquad
\tilde{t}_2^{\,(j)} = \bar{t}_2-\frac{\rho_2^{(j)}}{c},
\]
where $c$ is the speed of light and $\rho_2^{(j)}$ is the value of
$\rho_2$ for the $j$-th solution.
%Note that the value of $\rho_2^{(j)}$ (and then $\tilde{t}_2$) is
%usually different for different solutions.
% *** use $\rho_2^{(j)}$ ? ***.

We emphasize that the accepted solutions of the system are such that
the Keplerian elements at the two epochs are the same. In fact, from
$\angmom_1-\angmom_2=\bzero$ we get
\[
i_1=i_2,\qquad
\Omega_1=\Omega_2,\qquad
a_1(1-e_1^2)=a_2(1-e_2^2),
\]
where $i$, $\Omega$, $a$ and $e$ denote the inclination, ascending
node, semi-major axis and eccentricity. In particular, from
$z_2^2|\erre_2|^2 - \mu^2 = 0$, which is a consequence of
system~\eqref{redsys}, we get $z_2 = \pm\frac{\mu}{|\erre_2|}$. Having
discarded the solutions with $z_2\le 0$, we obtain
\[
\mu(\lenz_1-\widetilde{\lenz}_2)=\bzero
\quad\Leftrightarrow\quad
\mu(\lenz_1-\lenz_2)=\bzero,
\]
which implies
\[
e_1=e_2,\qquad
\omega_1=\omega_2.
\]
Since the eccentrities $e_1$, $e_2$ are equal, then also
\[
a_1=a_2.
\]
%Moreover, being equation $z_2^2|\erre_2|^2 = \mu^2$ a consequence of
%system \eqref{redsys}, relation $z_2 = \frac{\mu}{|\erre_2 |}$ is
%automatically fulfilled.
We observe that for the other orbit determination methods using the
Keplerian integrals \citep{gdm10, gfd11, gbm15} the trajectories of
the solutions at the two epochs
%$\tilde{t}_1, \tilde{t}_2$
are not necessarily the same.

In the following subsections we explain how to choose or discard the
remaining solutions.

%----------------------------------------
\subsection{Selection without covariance}
\label{s:without}

In case of multiple solutions, labeled with $(j)$, we make our choice
according to the following procedure. We propagate each of the
computed orbits, referring to epoch $\tilde{t}_2^{\,(j)}$, backward to
epoch $\tilde{t}_1$ in the framework of Kepler's dynamics, and
consider the norm of the differences $\bm{\rho}_1^{(j)} - {\cal P}_1$,
where
\[
\bm{\rho}_1^{(j)}=\erre^{(j)}(\tilde{t}_1)-\bq(t_1)
\]
is the topocentric position of the $j$-th solution propagated to epoch
$\tilde{t}_1$. Note that
%\[
%\tilde{t}_1 = t_1 - \frac{\rho_1}{c}
%\]
%therefore
in general we have
\[
|\bm{\rho}_1^{(j)}|\neq\rho_1.
\]
We select the solution attaining the minimum value
%of this norm, i.e.
\[
\min_j|\bm{\rho}_1^{(j)}-{\cal P}_1|.
\]

%--------------------------------------
\subsection{Selection using covariance}
Assume that the data
\[
{\cal D}=({\cal P}_1,\Att_2),
\]
where
\[
{\cal P}_1=(\alpha_1,\delta_1,\rho_1),\qquad
\Att_2=(\alpha_2,\delta_2,\alphadot_2,\deltadot_2),
\]
have a covariance matrix
\[
\Gamma_{{\cal D}} =
\left[
\begin{array}{cc}
  \Gamma_{{\cal P}_1} & 0\cr
  0 &\Gamma_{\Att_2}\cr
\end{array}
\right].
\]
Let
\[
\Svec = \Svec({\cal D}) = ({\cal V}_1({\cal D}),{\cal R}_2({\cal D})),
\]
with
\[
{\cal V}_1 = (\xi_1,\zeta_1,\rhodot_1),\qquad {\cal R}_2 = (\rho_2,\rhodot_2)
\]
be a solution of
\begin{equation}
  \def\arraystretch{1.4}
  \bPhi(\Svec; \Avec) =
  \left(
  \begin{array}{c}
    \angmom_1-\angmom_2\cr
    \left[\mu\lenz_1 + (\erredot_2\cdot\erre_2)\erredot_2 \right]\cdot\bq_2\times\erre_2\cr
    \left[-(\erredot_1\cdot\erre_1)\erredot_1 - (\frac{1}{2}|\erredot_2|^2 + \energy_1)\erre_2
    + (\erredot_2\cdot\erre_2)\erredot_2\right]\cdot \erre_1\times(\erre_2-\bq_2)
  \end{array}
  \right) = \bzero.
  \label{sistema}
\end{equation}
Note that, generically, the solutions $\Svec$ of
system~\eqref{sistema} correspond to the $\Svec$ components of the
solutions of system~\eqref{polysys}.

%We consider the map
%\[
%\bm{\Phi}({\cal R},{\cal A}) = (\angmom_1-\angmom_2,
%\mu(\lenz_1-\lenz_2)\cdot\DD_2, \mu(\lenz_1-\lenz_2)\cdot(\erre_1\times(\erre_2-\bq_2)))
%\]

%Hint! we use $\erre_2-\bq_2$ in place of $\erho_2$ to compute the
%derivatives of $\bm{\Phi}$ with respect to
%$\erre_1,\erredot_1,\erre_2,\erredot_2$

Let $\mathbf{E}_{car}=(\mathbf{E}_{car}^{(1)},\mathbf{E}_{car}
^{(2)})$ and $\mathbf{E}_{att}=(\mathbf{E}_{att}^{(1)},
\mathbf{E}_{att}^{(2)})$ be the vectors of the Cartesian coordinates
and the attributable elements at epochs $\tilde{t}_1$ and
$\tilde{t}_2$.\footnote{for simplicity, in this section we drop the
  label $(j)$, referring to the possible multiple solutions.} Let us
introduce the transformation $\mathcal{T}_{att}^{car}:
\mathbf{E}_{att} \to \mathbf{E}_{car}$
by~(\ref{position}),~(\ref{velocity}) for both epochs, and consider
the map $\bPsi$ defined by
$\bPhi=\bPsi\circ\mathcal{T}_{att}^{car}$. Equation $\bPhi=\mathbf{0}$
is equivalent to the system %\textcolor{red}{(\ref{linksys})}.
\[
\angmom_1-\angmom_2=\bm{0},\quad
\mu(\lenz_1-\tilde{\lenz}_2)\cdot\DD_2=0,\quad
\mu(\lenz_1-\tilde{\lenz}_2)\cdot\erre_1\times(\erre_2-\bq_2)=0.
\]

We introduce the vector
\[
%\Svectil = (\alphadot_1,\deltadot_1,\rhodot_1).
\Vvectil_1 = (\alphadot_1,\deltadot_1,\rhodot_1).
\]
\noindent The covariance matrix of the Cartesian coordinates at epoch
$\tilde{t}_1$ is
\[
\Gamma_{car}^{(1)} = \frac{\partial\mathbf{E}_{car}^{(1)}}{\partial{\Avec}}
\Gamma_{\Avec}\left[\frac{\partial\mathbf{E}_{car}^{(1)}}{\partial{\Avec}}\right]^T\,,
\]
with
\[
\frac{\partial\mathbf{E}_{car}^{(1)}}{\partial{\Avec}} =
\frac{\partial\mathbf{E}_{car}^{(1)}}{\partial\mathbf{E}_{att}^{(1)}}
\frac{\partial\mathbf{E}_{att}^{(1)}}{\partial{\Avec}},
\qquad
\frac{\partial\mathbf{E}_{att}^{(1)}}{\partial{\Avec}} =
M\left[
\begin{array}{c}
  \begin{array}{cc}
    I_3 & O_{3\times 4}\cr
  \end{array}\cr
  %\displaystyle\frac{\partial(\alphadot_1,\deltadot_1,\rhodot_1)}{\partial{\Avec}}\cr
  \displaystyle\frac{\partial\Vvectil_1}{\partial{\Avec}}\cr
\end{array}
\right],
\]
where
\[
M = \left[
\begin{array}{ccccccc}
  1 & 0 & 0 & 0 & 0 & 0\cr
  0 & 1 & 0 & 0 & 0 & 0\cr
  0 & 0 & 0 & 1 & 0 & 0\cr
  0 & 0 & 0 & 0 & 1 & 0\cr
  0 & 0 & 1 & 0 & 0 & 0\cr
  0 & 0 & 0 & 0 & 0 & 1\cr
\end{array}
\right]
\]
is the matrix that exchanges the 3-rd, 4-th and 5-th lines, and
\[
\frac{\partial \Vvectil_1}{\partial{\Avec}} =
\frac{\partial\Vvectil_1}{\partial{\cal V}_1}
\frac{\partial{\cal V}_1}{\partial{\Avec}},
\qquad
\frac{\partial\Vvectil_1}{\partial{\cal V}_1} =
\left[
\begin{array}{ccc}
  \frac{1}{\rho_1\cos\delta_1} & 0 & 0\cr
  0 & \frac{1}{\rho_1} & 0\cr
  0 & 0 & 1\cr
\end{array}
\right].
\]

From the implicit function theorem, we have
\[
\frac{\partial{\Svectil}}{\partial{\Avec}}(\Avec) =
-\left[\frac{\partial\bPhi}{\partial{\Svectil}}(\mathbf{E}_{att})\right]^{-1}
\frac{\partial\bPhi}{\partial{\Avec}}(\mathbf{E}_{att}),
\]
\[
\frac{\partial\bPhi}{\partial{\Svectil}} = 
\left(\frac{\partial\bPsi}{\partial\mathbf{E}_{car}} 
\circ
{\cal T}_{att}^{car}\right)\frac{\partial{\cal T}_{att}^{car}}{\partial{\Svectil}},
\qquad
\frac{\partial\bPhi}{\partial{\Avec}} = 
\left(\frac{\partial\bPsi}{\partial\mathbf{E}_{car}} 
\circ
{\cal T}_{att}^{car}\right)\frac{\partial{\cal T}_{att}^{car}}{\partial{\Avec}},
\]
with
\[
\Svectil = (\Vvectil_1,{\cal R}_2),\qquad
\frac{\partial\Svectil}{\partial \Svec} =
\left[
\begin{array}{cc}
  \frac{\partial\Vvectil_1}{\partial {\cal V}_1} & O_{3\times 2}\cr
  O_{2\times 3} & I_{{2}}\cr
\end{array}
\right].
\]
The matrices ${\partial{\cal T}_{att}^{car}}/{\partial{\Svectil}}$ and
${\partial{\cal T}_{att}^{car}}/{\partial{\Avec}}$ are respectively
made by columns 3,~4,~6,~11,~12 and by columns 1,~2,~5,~7,~8,~9,~10 of
${\partial\mathbf{E}_{car}}/{\partial\mathbf{E}_{att}}$.

The covariance matrix of the Cartesian coordinates at epoch
$\tilde{t}_2$ is given by
\[
\Gamma_{car}^{(2)} = \frac{\partial\mathbf{E}_{car}^{(2)}}{\partial\mathcal{D}}\Gamma_{\mathcal{D}}
\left[\frac{\partial\mathbf{E}_{car}^{(2)}}{\partial\mathcal{D}} \right]^T,
\]
with
\[
\def\arraystretch{1.2}
\frac{\partial\mathbf{E}_{car}^{(2)}}{\partial\mathcal{D}} =
\frac{\partial\mathbf{E}_{car}^{(2)}}{\partial\mathbf{E}_{att}^{(2)}}
\frac{\partial\mathbf{E}_{att}^{(2)}}{\partial\mathcal{D}},
\qquad
\frac{\partial\mathbf{E}_{att}^{(2)}}{\partial\mathcal{D}} =
\left[
\begin{array}{c}
  \begin{array}{cc}
    O_{4\times 3} & I_4\cr
  \end{array}\cr
  \frac{\partial\mathcal{R}_2}{\partial\mathcal{D}}
\end{array}
\right]
\]
and
\[
\frac{\partial(\rho_2,\rhodot_2)}{\partial\Svectil} =
\frac{\partial(\rho_2,\rhodot_2)}{\partial \Svectil}
\frac{\partial\Svectil}{\partial\mathcal{D}},\qquad
\frac{\partial(\rho_2,\rhodot_2)}{\partial\Svectil} =
\left[
\begin{array}{ccccc}
  0 & 0 & 0 & 1 & 0 \cr
  0 & 0 & 0 & 0 & 1 \cr
\end{array}
\right].
\]

\noindent For a given vector ${\bf u}\in\R^3$ we define the map
%\textbf{hat map}
\[
\R^3\ni (u_1,u_2,u_3) = {\bf u}\mapsto\widehat{\bf u}
%\stackrel{def}{=}
=
\left[
\begin{array}{ccc}
  0 & -u_3 & u_2\cr
  u_3 & 0 & -u_1\cr
  -u_2 & u_1 & 0\cr
\end{array}
\right]
\in so(3).
\]
Then, using $\hat{\bf u}^T=-\hat{\bf u}$, we have
\[
\frac{\partial\bPsi}{\partial\mathbf{E}_{car}} = 
\left[
\begin{array}{cccc}
  -\widehat{{\erredot}_1} & \widehat{{\erre}_1} & \widehat{{\erredot}_2} & -\widehat{{\erre}_2}\cr
  \stackrel{}{\displaystyle\frac{\partial\Phi_4}{\partial \erre_1}} & \displaystyle\frac{\partial\Phi_4}{\partial\erredot_1}
  & \displaystyle\frac{\partial\Phi_4}{\partial\erre_2} & \displaystyle\frac{\partial\Phi_4}{\partial\erredot_2}\cr
  \stackrel{}{\displaystyle\frac{\partial\Phi_5}{\partial\erre_1}} & \displaystyle\frac{\partial\Phi_5}{\partial\erredot_1}
  & \displaystyle\frac{\partial\Phi_5}{\partial\erre_2} & \displaystyle\frac{\partial\Phi_5}{\partial\erredot_2}\cr
\end{array}
\right],
\]
where
\begin{align*}
  \Phi_4 &= \left[\mu\lenz_1 + (\erredot_2\cdot\erre_2)\erredot_2\right]\cdot\bq_2\times\erre_2,\\[0.5ex]
  \Phi_5 &= \left[-(\erredot_1\cdot\erre_1)\erredot_1 - (\frac{1}{2}|\erredot_2|^2 + \energy_1)\erre_2
    + (\erredot_2\cdot\erre_2)\erredot_2\right]\cdot\erre_1\times(\erre_2-\bq_2),
\end{align*}
and
\begin{align*}
  & \frac{\partial\Phi_4}{\partial\erre_1} = \left(|\erredot_1|^2-\frac{\mu}{|\erre_1|}\right)(\bq_2\times\erre_2)
  - \erredot_1\left[\erredot_1\cdot(\bq_2\times\erre_2)\right] + \frac{\mu}{|\erre_1|^3}\erre_1\left[\erre_1\cdot(\bq_2\times\erre_2)\right],
  \\[0.5ex]
  & \frac{\partial\Phi_4}{\partial\erredot_1} = 2\erredot_1\left[\erre_1\cdot(\bq_2\times\erre_2)\right]
  - \erre_1\left[\erredot_1\cdot(\bq_2\times\erre_2)\right] - (\erredot_1\cdot\erre_1)(\bq_2\times\erre_2),
  \\[0.5ex]
  & \frac{\partial\Phi_4}{\partial\erre_2} = \erredot_2\left[\erredot_2\cdot(\bq_2\times\erre_2)\right]
  + \left[\left(|\erredot_1|^2-\frac{\mu}{|\erre_1|}\right)\erre_1-(\erredot_1\cdot\erre_1)\erredot_1
  +(\erredot_2\cdot\erre_2)\erredot_2\right]\times\bq_2,
  \\[0.5ex]
  & \frac{\partial\Phi_4}{\partial\erredot_2} = \erre_2\left[\erredot_2\cdot(\bq_2\times\erre_2)\right]
  + (\erredot_2\cdot\erre_2)(\bq_2\times\erre_2),
  \\[0.5ex]
  & \frac{\partial\Phi_5}{\partial \erre_1} = -\erredot_1\left[\erredot_1\cdot\erre_1\times(\erre_2-\bq_2)\right]
  - \frac{\mu}{|\erre_1|^3}\erre_1\left[\erre_2\cdot\erre_1\times(\erre_2-\bq_2)\right]\cr
  & \quad\qquad + (\erre_2-\bq_2)\times\left[-(\erredot_1\cdot\erre_1)\erredot_1-\left(\frac{|\erredot_2|^2}{2}-\frac{|\erredot_1|^2}{2}
  +\frac{\mu}{|\erre_1|}\right)\erre_2+(\erredot_2\cdot\erre_2)\erredot_2\right],
  \\[0.5ex]
  & \frac{\partial\Phi_5}{\partial\erredot_1} = -\erre_1\left[\erredot_1\cdot\erre_1\times(\erre_2-\bq_2)\right]
  - (\erredot_1\cdot\erre_1)\left[\erre_1\times(\erre_2-\bq_2)\right] + \erredot_1\left[\erre_2\cdot\erre_1\times(\erre_2-\bq_2)\right],
  \\[0.5ex]
  & \frac{\partial\Phi_5}{\partial\erre_2} = -\left(\frac{|\erredot_2|^2}{2}+\frac{|\erredot_1|^2}{2}-\frac{\mu}{|\erre_1|}\right)
  \left[\erre_1\times(\erre_2-\bq_2)\right] + \erredot_2\left[\erredot_2\cdot\erre_1\times(\erre_2-\bq_2)\right]\cr
  & \quad\qquad + \left[-(\erredot_1\cdot\erre_1)\erredot_1-\left(\frac{|\erredot_2|^2}{2}+\frac{|\erredot_1|^2}{2}-\frac{\mu}{|\erre_1|}\right)\erre_2
  +(\erredot_2\cdot\erre_2)\erredot_2\right]\times\erre_1,
  \\[0.5ex]
  & \frac{\partial\Phi_5}{\partial\erredot_2} = -\erredot_2\left[\erre_2\cdot\erre_1\times(\erre_2-\bq_2)\right]
  + \erre_2\left[\erredot_2\cdot\erre_1\times(\erre_2-\bq_2)\right] + (\erredot_2\cdot\erre_2)\left[\erre_1\times(\erre_2-\bq_2)\right].
\end{align*}
%in fact
%\[
%\frac{\partial (\erre_1\cdot\wvec)}{\partial\erre_2} = \q_2\times\erre_1\,,
%\hskip 1cm
%\frac{\partial (\erredot_1\cdot\wvec)}{\partial\erre_2} = \q_2\times\erredot_1\,,
%\hskip 1cm
%\frac{\partial (\erre_1\cdot\wvec)}{\partial\erre_2} = \q_2\times\erredot_2\ .
%\]

Equations~\eqref{redsys} do not set any constraints on the mean
anomalies $\ell_1$, $\ell_2$ of the pairs of Keplerian orbits that we
compute. To properly select a solution we integrate back each orbit
$\mathbf{E}_{\text{car}} ^{(2)}$ computed at epoch $\tilde{t}_2$, to
the epoch $\tilde{t}_1$ of the orbit $\mathbf{E}_{\text{car}} ^{(1)}$,
together with its covariance matrix $\Gamma ^{(2)}$. Then, we compute
a predicted position vector $\mathcal{P}_{1,p}$ and its $3\times 3$
covariance matrix $\Gamma_{\mathcal{P}_{1,p}}$. To define a metric to
choose/discard solutions we use a three-dimensional version of the
attribution algorithm \citep{mg10}, and introduce the identification
penalty
\[
\chi_3^2 = (\mathcal{P}_1-\mathcal{P}_{1,p})\cdot\left[C_{\mathcal{P}_{1,p}}-C_{\mathcal{P}_{1,p}}\Gamma_0C_{\mathcal{P}_{1,p}}\right]
(\mathcal{P}_1-\mathcal{P}_{1,p}),
\]
where
\[
C_{\mathcal{P}_{1,p}} = \Gamma_{\mathcal{P}_{1,p}}^{-1},\qquad
\Gamma_0 = C_0^{-1},
\]
with
\[
C_0 = C_{\mathcal{P}_{1,p}}+C_{\mathcal{P}_1},\qquad
C_{\mathcal{P}_1} = \Gamma_{\mathcal{P}_1}^{-1}.
\]
We keep only the solutions with the lowest values of $\chi_3$ or, if
we wish to admit multiple solutions, we keep the ones whose value of
$\chi_3$ is below a certain threshold.

%------------------------
%------------------------
\section{Numerical tests}
\label{s:numerics}

In this section we show an application of our algorithm to asteroids
undergoing a close approch with the Earth.  We consider all the
near-Earth asteroids present in the NEODyS database\footnote{\tt
  http://newton.spacedys.com/neodys} to the date of December 6, 2023,
and select the
%\footnote{\textcolor{red}{(ERICA) I think that we can drop 'the'}}
$1305$ asteroids with Earth MOID\footnote{Minimum Orbit Intersection
  Distance, see \citet{BM1994}, \citet{GBG2023}} $d_{\rm min}$ smaller
than $10^{-3}$ au. For each selected asteroid we compute the points
$P_A$ and $P_\oplus$, lying on the osculating orbits of the asteroid
and the Earth at epoch $t_0$, where the minimum distance $d_{\rm min}$
is attained. Then, we find the closest epoch $t_1>t_0$ when the Earth
arrives at $P_\oplus$ and change the asteroid phase so that it would
arrive at $P_A$ at the same epoch $t_1$ in the framework of the
Sun-asteroid two-body dynamics. In this way, we try to enhance the
effect of the close approach with the Earth at $t_1$. A full $n$-body
propagation, from $t_0$ to $t_1$, is used to compute the topocentric
position vector ${\cal P}_1$
% actually at \tilde{t}_1
from the point of view of the Pan-STARRS 1 telescope
\citep{panstarrs1}.
Continuing the propagation from $t_1$ to a successive epoch $t_2$, we
derive an attributable ${\cal A}_2$ by coordinate change.  As a
result, the components $\alpha_2$, $\delta_2$, $\alphadot_2$,
$\deltadot_2$ of ${\cal A}_2$ are not affected by interpolation
errors, which are usually introduced when an attributable is obtained
from the astrometric observations.

Assume we can model a close approach at epoch $t_1$ by an
instantaneous velocity change, like in \"Opik's theory
\citep{opik1976}. Then, our algorithm can be applied using ${\cal
  P}_1, {\cal A}_2$ as input data, because the values of the Keplerian
integrals are conserved in $(t_1,t_2]$. At $t_1$, when the velocity
instantaneously changes, the position vector ${\cal P}_1$ remains
unchanged.

In practice, the close encounter is not instantaneous, but takes a
finite time: we select the epoch $t_2$ when the close approach is
over, as described below.
% Tisserand's parameter and estimate of the close approach time
In a typical close encounter, the value of Tisserand's parameter $T$
changes quickly, going back approximately to its pre-encounter value
when the encounter is over. In \citet[][Sec. 3]{SOI2023}, for a given
value $C$ of the Jacobi constant, the authors provide a value $d_C$ of
the geocentric distance attaining $|\frac{dT}{dt}|\leq \varepsilon$,
for a fixed small parameter $\varepsilon$, so that we can think that
the encounter is over when the asteroid reaches that distance from the
Earth. Following this approach, we choose $t_2$ as the time when the
value of the distance becomes $d_C$.

\begin{figure}[t!]
  \begin{center}
    \includegraphics[width=7.3cm]{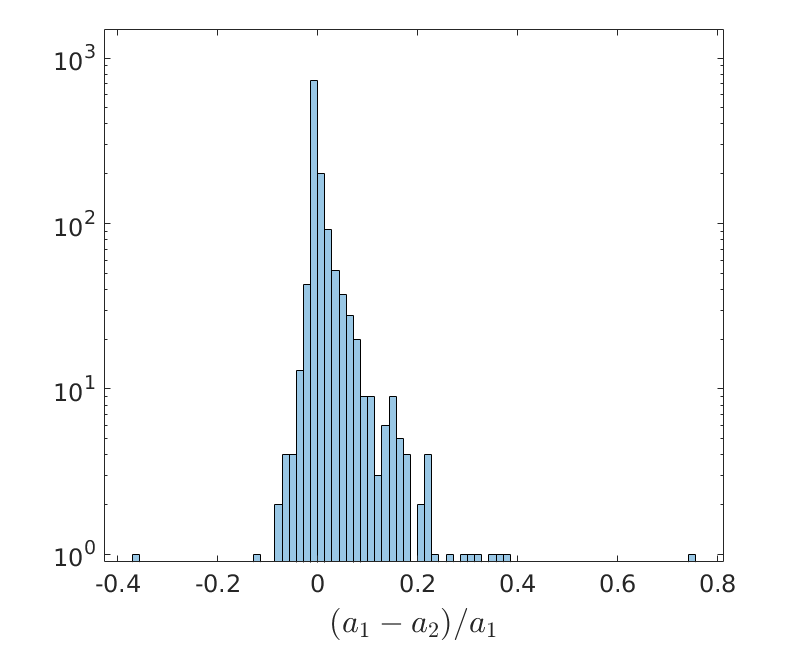}
    %\hskip 0.1cm
\includegraphics[width=7.3cm]{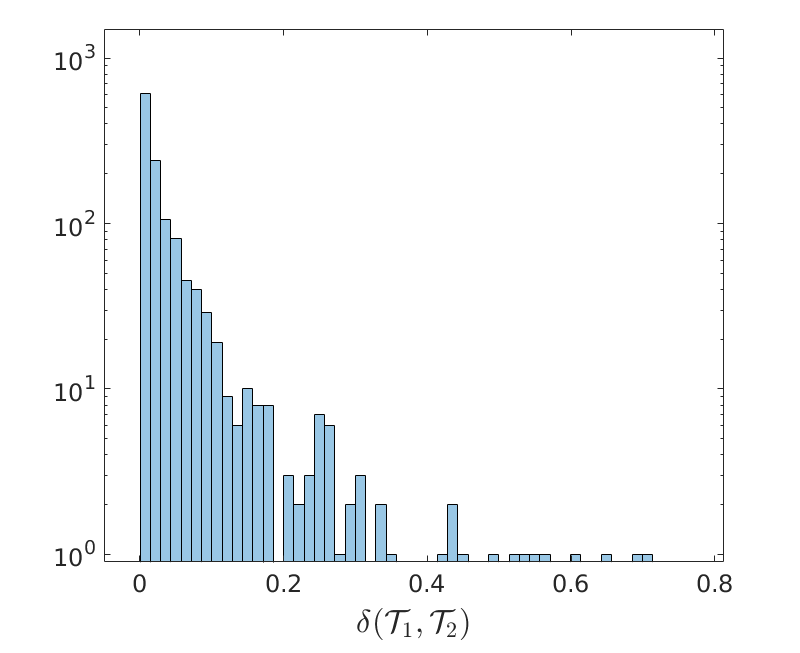}
\end{center}
    \caption{Effects of the close approach on the semimajor axes
      (left) and on the whole trajectories (right).}
\label{fig:ca_effect}
\end{figure}

In Fig.~\ref{fig:ca_effect} we show the effect of the close approach
on these orbits.  On the left, we plot the distribution of the
relative differences between the semimajor axes $a_1$, $a_2$ obtained
by the known orbits propagated at epochs $t_1$ and $t_2$ using the
software OrbFit\footnote{\tt http::adams.dm.unipi.it/orbfit}. On the
right, we plot the distribution of the distance $\delta({\cal
  T}_1,{\cal T}_2)$ between the propagated trajectories ${\cal T}_1$,
${\cal T}_2$, where
\begin{equation}
  \delta({\cal T}_1, {\cal T}_2) = \sqrt{\frac{(a_1-a_2)^2}{a_1^2} +
    (e_1-e_2)^2 + (i_1-i_2)^2 + (\Omega_1-\Omega_2)^2 + (\omega_1-\omega_2)^2 }.
  \label{deltaTra}
\end{equation}
In~\eqref{deltaTra} the subscripts $1$, $2$ of the Keplerian elements
$a$, $e$, $i$, $\Omega$, $\omega$ and the trajectories ${\cal T}$
refer to the epochs $t_1$, $t_2$.
%\footnote{\textcolor{red}{(ERICA) I would add 'respectively'.}}
In most cases the relative change in semimajor axis is within
20\%. However, there are a few cases where the change is larger, with
one extreme case (Fig~\ref{fig:ca_effect}, left)
%\footnote{\textcolor{red}{(ERICA) I would write 'top left panel'.}}
where it passes from $a_1\sim 2.84$ to $a_2\sim 0.69$ au. The
approximate values of the mean and the standard deviation of
$(a_1-a_2)/a_1$ are about 0.012 and 0.0478, respectively.
    
% differences between the computed orbits and the true ones
Next, we show the performace of our algorithm using ${\cal P}_1$,
${\cal A}_2$ as input data, which have been computed for each of these
orbits by a full $n$-body propagation with OrbFit.
% selection of the solutions
In case of multiple solutions, we select the best one according to the
procedure explained in Section~\ref{s:without}.

In 8 cases, out of 1305, we could not obtain an orbital
solution. However, by increasing the time span $[t_1,t_2]$ by $10\%$,
while keeping the same value of $t_1$, thus changing only ${\cal
  A}_2$, we recovered an orbital solution in all these cases.
%that in a few cases was very good.
%
In one case the computation of the second epoch $t_2$ with Newton's
method failed. However, employing the {\em starting guess} value for
$t_2$, computed with a geocentric two-body dynamics, we obtained an
orbital solution also in this case.

In Fig.~\ref{fig:diff_comp_known} we show the distribution of the
differences between the Keplerian elements computed with our algorithm
(with subscript $c$) and the same elements at $t_2$ obtained by
propagation. In Fig.~\ref{fig:diff_comp_known} (bottom right) we also
show the distribution of the distances $\delta( {\cal T}_2,{\cal
  T}_c)$ between propagated (${\cal T}_2$) and computed (${\cal T}_c$)
trajectories, where the function $\delta$ is the same as
in~\eqref{deltaTra}.

\begin{figure}[t!]
  \begin{center}
    \includegraphics[width=7.3cm]{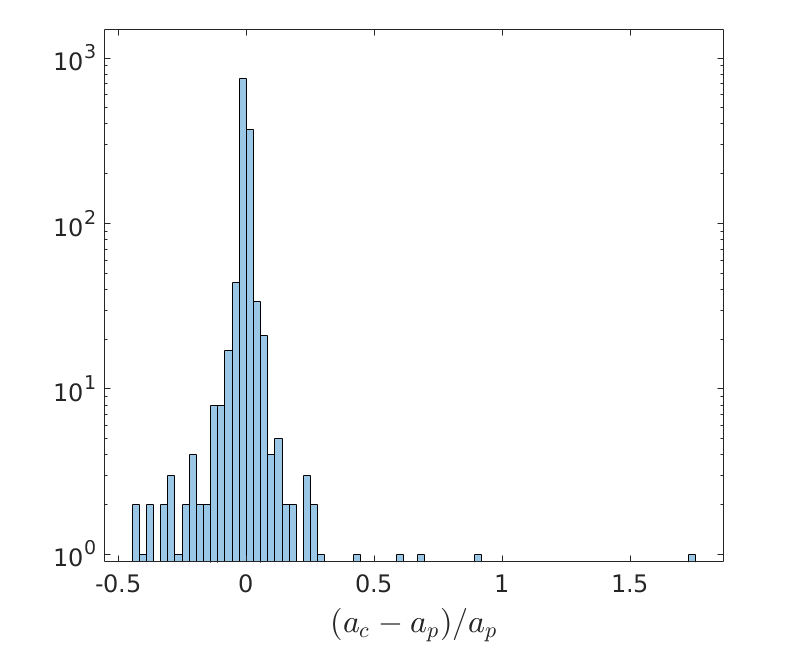}
    \includegraphics[width=7.3cm]{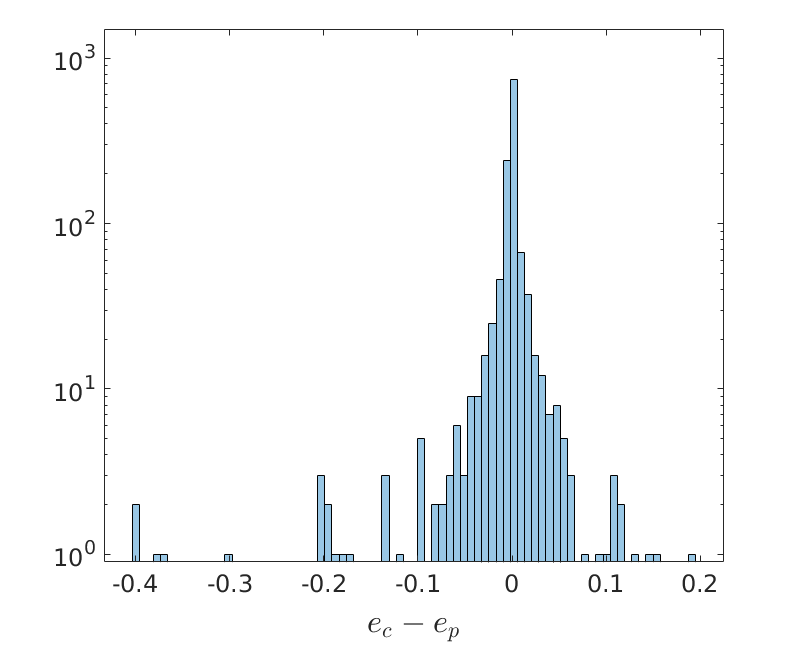}\\
    \includegraphics[width=7.3cm]{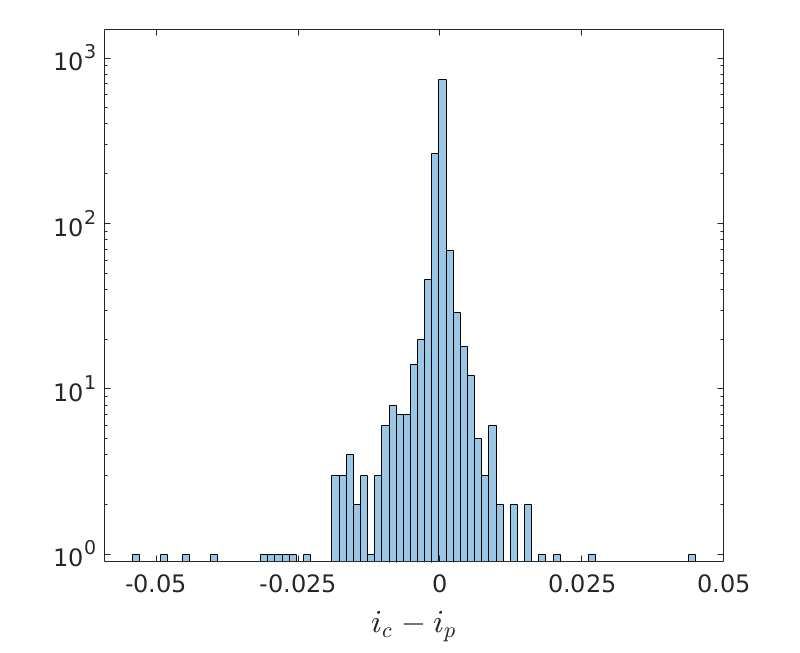}
    \includegraphics[width=7.3cm]{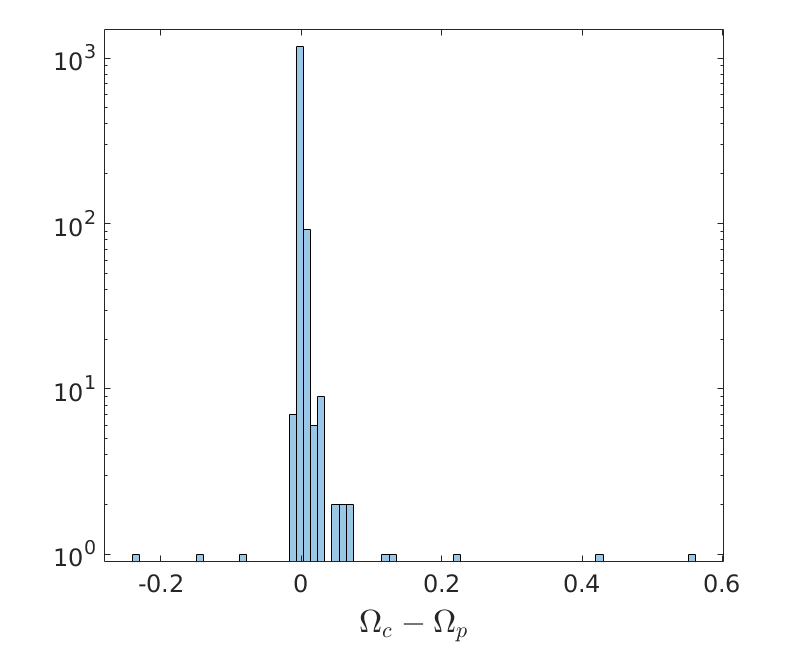}\\
    \includegraphics[width=7.3cm]{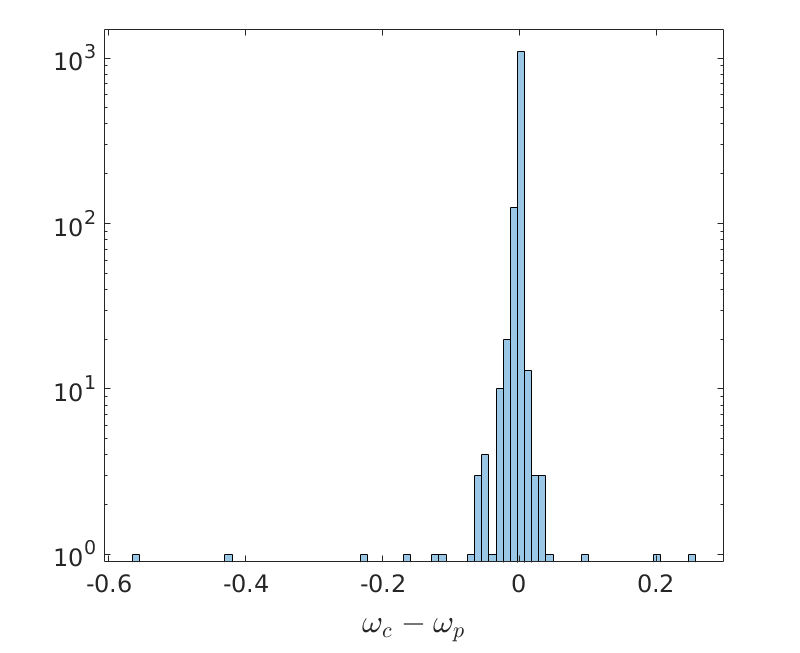}
    \includegraphics[width=7.3cm]{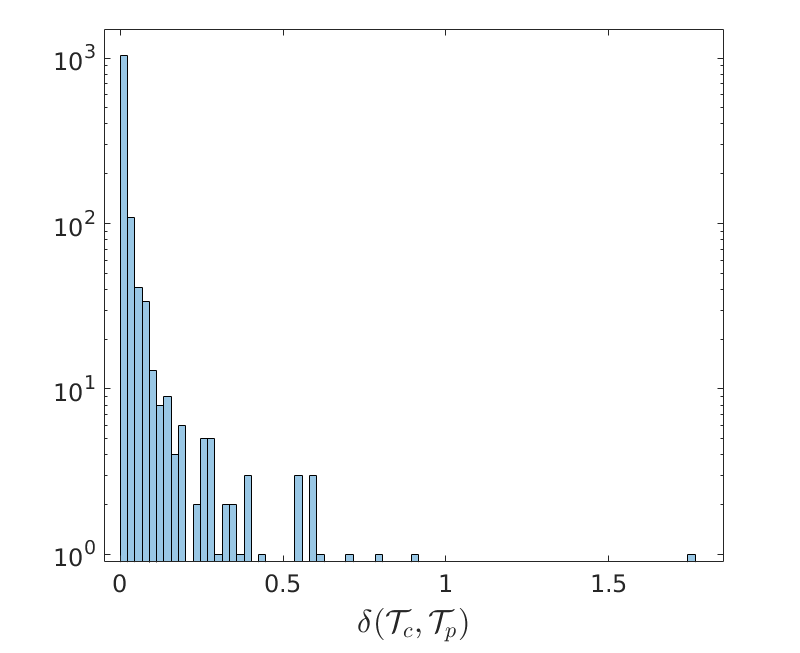}
  \end{center}
  \caption{Distribution of the differences in orbital elements $a$
    (top left), $e$ (top right), $i$ (center left), $\Omega$ (center
    right), $\omega$ (bottom left) between the propagated trajectory
    and the one computed with our algorithm at epoch
    $t_2$. Distribution of the distances $\delta({\cal T}_2,{\cal
      T}_c)$ between propagated and computed trajectories (bottom
    right).}
\label{fig:diff_comp_known}
\end{figure}

The results of this preliminary statistical test are satisfactory:
%apart from a few outliers,
the mean and the standard deviation of the distributions displayed in
Fig.~\ref{fig:diff_comp_known} are shown in Table~\ref{tab:meanstd}.

\renewcommand{\arraystretch}{1.3}
\begin{table}[h!]
  \begin{center}
    \begin{tabular}{c|c|c} 
      & mean & std \cr
      \hline
      $(a_c-a_2)/a_2$ &$-$5.0803$\times$10$^{-4}$ &0.0778 \cr
      $e_c-e_2$ &$-$0.003 &0.0341 \cr
      $i_c-i_2$ &$-$3.8397$\times$10$^{-4}$  &0.0045 \cr
      $\Omega_c-\Omega_2$ &0.0019 &0.0232 \cr
      $\omega_c-\omega_2$  &$-$0.0021 &0.0244 \cr
      \hline
      $\delta({\cal T}_c,{\cal T}_2)$  &0.0264 &0.0877 \cr
      \hline
    \end{tabular}
  \end{center}
  \caption{Mean and standard deviation of the distributions
    represented in Fig.~\ref{fig:diff_comp_known}.}
  \label{tab:meanstd}
\end{table}

\begin{figure}[h!]
  \centerline{\includegraphics[width=9.3cm]{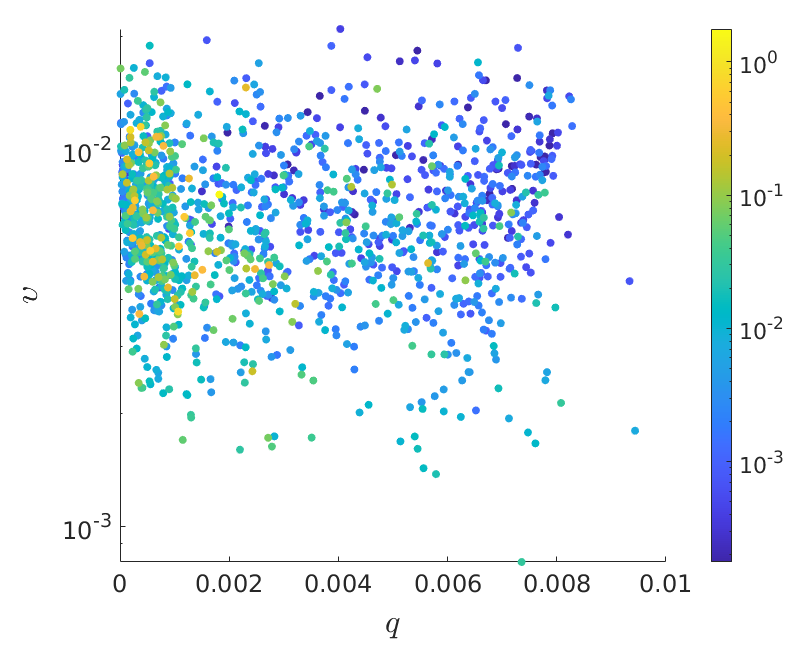}}
  \caption{Values of $q$ and $\upsilon$ in a log-log plot. The colors
    represent the values of the distance $\delta({\cal T}_2,{\cal
      T}_c)$ between known and computed trajectories.}
  \label{fig:q_vq_deltaTra}
\end{figure}

Finally, we show the dependence of the results of our algorithm on the
type of close encounter. Following \citet[][Sect. 3]{SOI2023}, close
encounters can be classified as {\em deep/shallow} and {\em
  fast/slow}, according to the values of the minimum distance $q$ from
the Earth center and the corresponding geocentric velocity $\upsilon$.
Here, we compute an approximation of $q$ and $\upsilon$, still denoted
by these symbols, using the following procedure: $q$ is the pericenter
distance of the geocentric two-body orbit computed from the Keplerian
elements at $t_1$, and $\upsilon$ is the corresponding two-body
velocity. In Fig.~\ref{fig:q_vq_deltaTra} we plot the values of $q$
and $\upsilon$ using a color scale representing the values of the
trajectory distance $\delta({\cal T}_c,{\cal T}_2)$.

For values of $q$ greater than $0.001$ au, we obtain better results
(e.g. lower values of the trajectory distance) for higher values of
the velocity $\upsilon$. For $q$ smaller than $0.001$ au, the results
appear worse, and there is no apparent correlation with velocity in
the figure.

%--------------------
%--------------------
\section{Conclusions}

We have introduced a method to compute preliminary orbits with one
topocentric position vector ${\cal P}_1$ and a very short arc of
optical observations, from which we can derive an attributable ${\cal
  A}_2$. This method is based on polynomial equations coming from the
first integrals of Kepler's problem, i.e. angular momentum, energy,
and Laplace-Lenz vector. Using the conservation laws of these
integrals, after introducing the auxiliary variable $z_2$, we obtain a
polynomial system that always has solutions, at least in the complex
field, even if ${\cal P}_1$, ${\cal A}_2$ do not correspond to the
same celestial object. There are some checks that can be performed to
accept or reject solutions.
%
% of course the solution must be real, and it must have positive
%values of $\rho_2$ and $z_2$. In the equations we do not impose
%anything about the time law, therefore if we have a solution
%compatible with the previous checks, we can try to assess its
%reliability by propagating the orbit obtained at the second epoch
%backward to the first epoch and comparing the topocentric position
%that we get with ${\cal P}_1$.

We applied this algorithm to the computation of $1305$ NEA orbits,
whose phase was changed in order to enhance the close encounter effect
with the Earth: these preliminary results are satisfactory, see
Section~\ref{s:numerics}.
The ideal situation for the application of the proposed algorithm in
this context would be given by an instantaneous effect of the close
encounter, which is not the case.
The selected epoch $t_1$ does not exactly correspond to the time of
passage at the MOID because we use a full $n$-body propagation
starting from $t_0$, so that at $t_1$ both the osculating trajectory
and the time law along it may have changed.  Moreover, from this
preliminary test, we saw that also varying $t_2$ may affect the
results.
Therefore, a more detailed study is necessary to better understand the
applicability in case of close encounters and the reliability of the
computed orbits. On top of that, the sensitivity of the algorithm to
astrometric errors has still to be investigated.

Applications of the method introduced in this work to the orbit
computation of Earth satellites undergoing a maneuvre are also
possible, and are worth to be investigated in a future work.

%--------------------------
%--------------------------
\paragraph{Acknowledgments}
The authors acknowledge the project MIUR-PRIN 20178CJA2AB ``New
Frontiers of Celestial Mechanics: theory and
applications''. G.F. Gronchi and G. Ba\`u acknowledge the project
MSCA-ITN Stardust-R, Grant Agreement n. 813644 under the H2020
research and innovation program. E. Scantamburlo acknowledges the
project ``Advanced Space System Engineering to Address Broad Societal
Benefits – Starting Grant" funded by a contract between Politecnico di
Torino and Compagnia di San Paolo (CSP) 2019/2021 within the call
``Attrazione e retention di docenti di qualità''.

\appendix

%-----------------
%-----------------
\section{Appendix}
\label{app:q5q6}

\begin{lemma}
  For $\angmom_1 = \angmom_2$, the generators $\mathfrak{q}_4$,
  $\mathfrak{q}_5$, $\mathfrak{q}_6$ defined in~\eqref{qgenerators}
  assume the following form:
  \begin{equation*}
    \begin{split}
      \tilde{\mathfrak{q}}_4 &= (\angmom_1\cdot\erho_1)\left[\erre_1\cdot(\erredot_1-\erredot_2)\right]+(\erre_2\cdot\DD_1)z_2,\cr
      \tilde{\mathfrak{q}}_5 &= (\angmom_2\cdot\erho_2)\left[\erre_2\cdot(\erredot_1-\erredot_2)\right]-(\erre_1\cdot\DD_2)\frac{\mu}{|\erre_1|},\cr
      \tilde{\mathfrak{q}}_6 &= (\angmom_2\cdot\erho_2)\left[\erre_1\cdot(\erredot_1-\erredot_2)\right]-(\erre_1\cdot\DD_2)z_2.
    \end{split}
  \end{equation*}
\end{lemma}
\begin{proof}
  The result is given by a direct computation:
  \begin{equation*}
    \begin{split}
      \tilde{\mathfrak{q}}_4 &= -(\erredot_1\cdot\erre_1)\left[\erredot_1\cdot(\erre_1\times\erho_1)\right]-\mu\tilde{\lenz}_2\cdot\DD_1\cr
      &= (\angmom_1\cdot\erho_1)(\erredot_1\cdot\erre_1)-\mu\tilde{\lenz}_2\cdot\DD_1\cr
      &= (\angmom_1\cdot\erho_1)(\erredot_1\cdot\erre_1)-(\|\erre_2\|^2-z_2)\left[\erre_2\cdot(\erre_1\times\erho_1)\right]
      -(\erredot_2\cdot\erre_2)\left[\erredot_2\cdot(\erre_1\times\erho_1)\right]\cr
      &= (\angmom_1\cdot\erho_1)(\erredot_1\cdot\erre_1)-(\erredot_2\times\angmom_2)\cdot(\erre_1\times\erho_1)+(\erre_2\cdot\DD_1)z_2\cr
      &= (\angmom_1\cdot\erho_1)(\erredot_1\cdot\erre_1)-\left[\erredot_2\times(\erre_1\times\erredot_1)\right]
      \cdot(\erre_1\times\erho_1)+(\erre_2\cdot\DD_1)z_2\cr
      &= (\angmom_1\cdot\erho_1)(\erredot_1\cdot\erre_1)-\left[\erre_1(\erredot_1\cdot\erredot_2)-\erredot_1(\erredot_2\cdot\erre_1)\right]
      \cdot(\erre_1\times\erho_1)+(\erre_2\cdot\DD_1)z_2\cr
      &= (\angmom_1\cdot\erho_1)(\erredot_1\cdot\erre_1)+(\erredot_2\cdot\erre_1)\left[\erredot_1\cdot(\erre_1\times\erho_1)\right]
      +(\erre_2\cdot\DD_1)z_2\cr
      &= (\angmom_1\cdot\erho_1)\left[\erre_1\cdot(\erredot_1-\erredot_2)\right]+(\erre_2\cdot\DD_1)z_2,
    \end{split}
  \end{equation*} 
  \begin{equation*}
    \begin{split}
      \tilde{\mathfrak{q}}_5 &= \left(\erredot_1\times\angmom_1-\frac{\mu}{|\erre_1|}\right)\cdot(\erre_2\times\erho_2)
      -(\angmom_2\cdot\erho_2)(\erredot_2\cdot\erre_2)\cr
      &= \left(\erredot_1\times\angmom_2-\frac{\mu}{|\erre_1|}\right)\cdot(\erre_2\times\erho_2)-(\angmom_2\cdot\erho_2)(\erredot_2\cdot\erre_2)\cr
      &= \left((\erredot_1\cdot\erredot_2)\erre_2-(\erre_2\cdot \erredot_1)\erredot_2-\frac{\mu}{|\erre_1|}\right)\cdot (\erre_2\times\erho_2)
      -(\angmom_2\cdot\erho_2)(\erredot_2\cdot\erre_2)\cr
      &= -(\erre_2\cdot\erredot_1)\left[\erredot_2\cdot(\erre_2\times\erho_2)\right]-\frac{\mu}{|\erre_1|}\erre_1\cdot(\erre_2\times\erho_2)
      -(\angmom_2\cdot\erho_2)(\erredot_2\cdot\erre_2)\cr
      & = (\angmom_2\cdot\erho_2)\left[\erre_2\cdot(\erredot_1-\erredot_2)\right]-(\erre_1\cdot\DD_2)\frac{\mu}{|\erre_1|},
    \end{split}
  \end{equation*}
  \begin{equation*}
    \begin{split}
      \tilde{\mathfrak{q}}_6 &= -\left[(\erredot_2\times\angmom_2)-z_2\erre_2\right]\cdot(\erre_1\times\erho_2)
      +(\angmom_2\cdot\erho_2)(\erredot_1\cdot\erre_1)\cr
      &= -\left[(\erredot_2\times\angmom_1)-z_2\erre_2\right]\cdot(\erre_1\times\erho_2)+(\angmom_2\cdot\erho_2)(\erredot_1\cdot\erre_1)\cr
      &= -\left[(\erredot_1\cdot\erredot_2)\erre_1-(\erredot_2\cdot\erre_1)\erredot_1-z_2\erre_2\right]\cdot(\erre_1\times\erho_2)
      +(\angmom_2\cdot\erho_2)(\erredot_1\cdot\erre_1)\cr
      &= (\erredot_2\cdot\erre_1)\left[\erredot_1\cdot(\erre_1\times\erho_2)\right]+z_2\left[\erre_2\cdot(\erre_1\times\erho_2)\right]
      + (\angmom_2\cdot\erho_2)(\erredot_1\cdot\erre_1)\cr
      &= (\angmom_2\cdot\erho_2)\left[\erre_1\cdot(\erredot_1-\erredot_2)\right]-(\erre_1\cdot\DD_2)z_2.
    \end{split}
  \end{equation*}
\end{proof}

\bibliography{mybib}{}
\bibliographystyle{plainnat}

\end{document}